\documentclass[11pt,final]{amsart}
\usepackage{amsmath, amsfonts, amsthm, amssymb,mathtools}
\usepackage{subfigure}
\usepackage{stmaryrd}
\usepackage{verbatim}
\usepackage{hyperref}
\usepackage{color}
\usepackage{ulem}
\usepackage[dvips]{graphicx}
\usepackage{enumerate}

\usepackage{showkeys}
\usepackage[marginclue, inline]{fixme}

\numberwithin{equation}{section}
\newtheorem{theorem}{Theorem}[section]
\newtheorem{corollary}[theorem]{Corollary}
\newtheorem{lemma}[theorem]{Lemma}
\newtheorem{proposition}[theorem]{Proposition}

\theoremstyle{definition}
\newtheorem{definition}[theorem]{Definition}
\newtheorem{remark}[theorem]{Remark}
\newtheorem{example}[theorem]{Example}

\newcommand{\ind}{1\hspace{-2.1mm}{1}} 

\newcommand{\RR}{\mathbb{R}}
\newcommand{\PP}{\mathbb{P}}
\newcommand{\EE}{\mathbb{E}}
\newcommand{\Dd}{\mathcal{D}}
\newcommand{\Ii}{\mathcal{I}}
\newcommand{\Pp}{\mathcal{P}}
\newcommand{\Ll}{\mathcal{L}}
\newcommand{\D}{\mathrm{d}}
\newcommand{\eps}{\varepsilon}

\newcommand{\E}{\mathrm{e}}

\newcommand{\QQ}{\mathbb{Q}}

\newcommand{\Nn}{\mathcal{N}}
\newcommand{\Ff}{\mathcal{F}}
\newcommand{\mm}{\mathfrak{m}}
\newcommand{\nn}{\mathfrak{n}}

\newcommand{\BS}{\mathrm{BS}}

\begin{document}

\title{Implied volatility in strict local martingale models}

\author{Antoine Jacquier}
\address{Department of Mathematics, Imperial College London}
\email{a.jacquier@imperial.ac.uk}
\author{Martin Keller-Ressel}
\address{TU Dresden, Fachrichtung Mathematik, Institut f\"ur Mathematische Stochastik}
\email{Martin.Keller-Ressel@tu-dresden.de}
\date{\today}
\thanks{AJ acknowledges financial support from the EPSRC First Grant EP/M008436/1. 
MKR acknowledges funding from the Excellence Initiative of the German Research Foundation (DFG) under grant ZUK 64.}

\maketitle
\begin{abstract}We consider implied volatilities in asset pricing models, where the discounted underlying is a strict local martingale under the pricing measure. Our main result gives an asymptotic expansion of the right wing of the implied volatility smile and shows that the strict local martingale property can be determined from this expansion. This result complements the well-known asymptotic results of Lee and Benaim-Friz, which apply only to true martingales. This also shows that `price bubbles' in the sense of strict local martingale behaviour can in principle be detected by an analysis of implied volatility. Finally we relate our results to left-wing expansions of implied volatilities in models with mass at zero by a duality method based on an absolutely continuous measure change.
\end{abstract}
\tableofcontents


\section{Introduction}

Most models in mathematical finance are bound to the paradigm that discounted asset prices are true martingales under the risk neutral measure. However, no-arbitrage theory 
as developed by Delbaen and Schachermayer~\cite{DelScha94, DelScha2, DelSchaBook} actually allows for the more general case of discounted asset prices being local martingales (or even sigma martingales) while still remaining consistent with absence of arbitrage in the sense of No Free Lunch with Vanishing Risk. Processes which fall in the latter category, but not in the former, i.e., which are local martingales but not true martingales are usually called strict local martingales. Although not trivial to construct, the appearance of strict local martingales in certain local and stochastic volatility models has been noted e.g. in~\cite{Jourdain, Lewis, Sin, DelScha}. 
In contrast to true martingale models, where the current asset price equals the risk-neutral expectation of the future discounted asset price (`fundamental price'), these two values differ from each other in strict local martingale models. For this very reason, strict local martingale models have been interpreted as models for markets with price bubbles, see e.g.~\cite{WillardBubbles, HestonBubbles} or the review paper~\cite{ProtterReview} for a detailed discussion.
For one-dimensional It\^o diffusions without drift, a precise characterisation of strict local martingales was obtained in~\cite{BleiEngelbert, DelShi, EngelbertSenf};
this was recently used in a series of papers by Jarrow, Kchia and Protter~\cite{Jarrow, Kchia1, Kchia2} to 
propose a statistical test based on realised volatility to determine whether a given underlying (LinkedIn's stock and gold) exhibits a price bubble.
In the context of option pricing, it has been shown that the strict local martingale property leads to unexpected and counter-intuitive behaviour. In particular Put-Call-parity fails, classical no-static-arbitrage bounds are no longer valid and there is some ambiguity about the proper valuation of derivatives with unbounded payoffs, such as Calls. 
We refer to~\cite{CoxHob} and a more elaborate discussion in section~\ref{Sub:option}.

In this paper we focus on the properties of implied volatilities and of the resulting implied volatility surface in asset pricing models with the strict local martingale property. With the exception of~\cite{Tehranchi}, 
who discusses long-time behaviour of implied volatilities, this is to our knowledge the first paper to discuss implied volatilities without imposing the true martingale assumption on the underlying. Due to the failure of Put-Call-parity and no-static-arbitrage bounds, it turns out that Put- and Call-implied volatility has to be distinguished and that Call-implied volatility does not always exist, see Theorem~\ref{thm:ExistenceIV}. 
Following an idea of Cox and Hobson~\cite{CoxHob}, we introduce the `fully collateralised Call', 
which restores Put-Call-parity and leads again to equality of Put-implied and Call-implied volatilities. 
Nevertheless, our main result, Theorem~\ref{thm:IVExpansion} shows that the right wing of the implied volatility smile in strict local martingale models exhibits an asymptotic behaviour which is fundamentally different from true martingale models. This also shows that in principle, strict local martingale models can be distinguished from true martingale models by analysing the implied volatility surface. Moreover, this result complements the well-known results of Lee~\cite{Lee} and Benaim and Friz~\cite{BenaimFriz} on the behaviour of the wings of the implied volatility smile in true (non-negative) martingale models.

Certain refinements of this main result are then obtained in Corollary~\ref{cor:dual}, by applying a duality method based on an absolutely continuous measure change. 
This method is well-known in the context of strict local martingales (see e.g.~\cite{CarrRuf, DelScha, KardNik, Pal, Ruf}) and puts positive strict local martingale models in duality with true martingale models with (positive) probability mass at zero. 
Recently, De Marco, Jacquier and Hillairet~\cite{DMHJ} and Gulisashvili~\cite{Guli} showed that, 
for a true martingale with mass at zero, the left tail of the smile was fully determined (up to second order) 
by the probability weight of this very mass.
Applying the duality method, the results on right-wing asymptotics of strict local martingales in this paper can be seen as a direct analogue of the left-wing asymptotics of true martingales with mass at zero.

\section{Preliminaries}

\subsection{Market models based on strict local martingales and stock price bubbles}\label{sub:market}
Let $(\Omega, \Ff, (\Ff_t)_{t\geq 0}, \QQ)$ be a probability space with a filtration satisfying the usual conditions and let 
$(S_t)_{t \geq 0}$ be a c\`adl\`ag non-negative local $\QQ$-martingale starting at $S_0 = 1$. 
As a non-negative local martingale, $S$ is also a supermartingale. Hence $\EE^{\QQ}[S_t]$ exists, 
is  bounded by $1$ and is a decreasing function of $t \geq 0$. We set
\begin{equation}\label{eq:mg_defect}
\mm_t := 1 - \EE^{\QQ}[S_t],
\end{equation}
and call $\mm_t$ the martingale defect (at $t \geq 0$) of~$S$. 
Clearly, $t \mapsto \mm_t$ is increasing, 
takes values in $[0,1]$ and $S$ is a true martingale on $[0,T]$ if and only if $\mm_T = 0$. We will be mainly interested in the complementary case $\mm_T > 0$, in which $S$ is a local martingale on $[0,T]$, but not a true martingale, i.e., a \emph{strict local martingale}. 
The size of~$\mm_T$ quantifies the difference to being a true martingale, hence the term `martingale defect'. The other boundary case $\mm_T = 1$ takes place if and only $S_T = 0$, $\QQ$ almost surely. 
We exclude this degenerate case from our analysis and work under the assumption that $\mm_t < 1$ for all $t \in [0,T]$.
We interpret $S$ as the price of a stock and $\QQ$ as a given pricing measure which determines the prices of derivative contracts. 
The historical/statistical measure will not play a role in our analysis and we do not make a priori assumptions on market completeness. 
The normalisation of~$S_0$ to~$1$ and the absence of discounting serve to simplify notation and arguments. 
Starting values $S_0 \neq 1$ can be accommodated by simple scaling, and discounting by interpreting~$S$ as a forward price 
and by adjusting option strikes to forward strikes. 

Recall that the market model $(S,\QQ)$ does not allow for arbitrage opportunities in the sense of No Free Lunch with Vanishing Risk (NFLVR). 
Indeed, the first fundamental theorem of asset pricing states that the NFLVR condition holds 
if and only if there exists an equivalent probability measure under which the (discounted) semimartingale stock price 
is a sigma martingale\footnote{If the stock price is locally bounded (in particular if it is continuous or has uniformly bounded jumps) 
then `sigma martingale' can be replaced by `local martingale'.}. 
As all local martingales are sigma martingales, 
the strict local martingale property of~$S$ under $\QQ$ implies NFLVR. 
We refer the interested reader to~\cite{DelScha2} or~\cite{DelSchaBook} for a precise definition of the NFLVR condition and details of the first fundamental theorem of asset pricing.

Concrete examples of market models where $\mm_T > 0$, i.e., where $S$ is a strict local $\QQ$-martingale appear in both local volatility and stochastic volatility models~\cite{Jourdain, Lewis, Sin}.  With regards to option pricing, strict local martingale models can exhibit quite unexpected behaviour: 
as we will discuss in more detail in section~\ref{Sub:option} below, Put-Call-parity fails and Put and Call prices may violate the classic no-static-arbitrage bounds in such models. Moreover, there exist strikes for which the European Call is strictly cheaper than its American counterpart and the European lookback Call options with payoff $\max_{0\leq t\leq T}S_t - K$ has infinite values for all $K\geq 0$, see~\cite{CoxHob, HestonBubbles}. 

Finally, note that from~\eqref{eq:mg_defect} the martingale defect $\mm_t$ can be interpreted as the difference between the current asset price $S_0 = 1$ and the risk-neutral expectation of its future value $S_t$ (its `fundamental value'). Following e.g.~\cite{HestonBubbles, ProtterReview} a non-zero difference therefore constitutes a \emph{price bubble} in which traded values of assets diverge from fundamental values; a phenomenon that cannot appear under a true martingale assumption on $S$.

\begin{example}\label{ex:diffusion}
Let~$X$ be the unique weak solution to the stochastic differential equation
$\D X_t = \sigma(X_t)\D W_t$, with $X_0=1$, 
where $\sigma$ is bounded away from zero and infinity on each compact subset of $(0,\infty)$, satisfies $\sigma(0)=0$,
and $W$ is a Brownian motion.
The behaviour of~$X$ at zero and its martingale property can be read directly 
from the following integrability conditions on $\sigma$ (cf.~\cite{DelShi}):
\begin{itemize}
\item $X_t>0$ for all $t\geq 0$ almost surely if and only if 
\begin{equation}\label{eq:integral1}
\int_{0}^{1}\frac{x}{\sigma^{2}(x)}\D x = \infty;
\end{equation}
\item $X$ is a strict local martingale if and only if
\begin{equation}\label{eq:integral2}
\int_{1}^{\infty}\frac{x}{\sigma^{2}(x)}\D x < \infty.
\end{equation}
\end{itemize}
As proved in~\cite{BleiEngelbert, EngelbertSenf} this classification remains valid under even weaker requirements on $\sigma$.
\end{example}

\subsection{Option pricing in strict local martingale models}\label{Sub:option}
Even though strict local martingale models are consistent with absence of arbitrage (in the sense of NFLVR), several apparent pathologies emerge in the context of option pricing, see e.g.~\cite{CoxHob} and~\cite{HestonBubbles} for a detailed discussion. We will summarise the phenomena which are relevant in the context of implied volatility. 
First let us define Call and Put prices on the underlying $S$ (and for fixed maturity $T$) by their risk-neutral expectations
\[
C_S(x) := \EE^{\QQ}(S_T - \E^x)_+ 
\qquad \text{and}\qquad
P_S(x) := \EE^{\QQ}(\E^x - S_T)_+,
\]
where in view of our analysis of implied volatility we parametrise by log-strike $x = \log K$. In complete markets these prices are indeed the unique minimal super-replication prices of their payoffs (see~\cite{CoxHob}), but---as will be seen below---other sensible prices for Puts and Calls which are consistent with absence of arbitrage exist. 
It is easy to see that 
\begin{equation}\label{eq:pc_parity}
C_S(x) - P_S(x) = 1 - \E^x - \mm_T
\end{equation}
holds for all real number~$x$, and hence Put-Call parity fails whenever $\mm_T > 0$, 
i.e., exactly when $S$ is a strict local martingale~\cite[Theorem~3.4]{CoxHob}. 
A further pathology emerges when we consider price bounds. 
A direct application of Jensen's inequality yields the inequalities
\begin{align}
\label{eq:call_bounds}
(1 - \mm_T - \E^x)_+ & \le C_S(x) < 1 - \mm_T, \\
\label{eq:put_bounds}
(\E^x - 1 + \mm_T)_+ & \le P_S(x) < \E^x,
\end{align}
for the Call and the Put, valid for all $x\in\RR$. 
For $\mm_T = 0$ the process $S$ is a true martingale and these bounds become the classic no-static-arbitrage bounds for Calls and Puts.\footnote{No-static-arbitrage refers to elementary static replication arguments that do no take into account admissibility of trading strategies. 
This notion has to be distinguished from no-arbitrage in the sense of NFLVR which does take into account admissibility of strategies.} 
In the strict local martingale case, $\mm_T > 0$, the lower bound for the Call falls 
outside the no-static-arbitrage region. 
The lower bound for the Put on the other hand increases with $\mm_T$ 
and thus always remains within the no-static-arbitrage region. 
Note that for any $x \in \RR$, both lower bounds can be attained, which follows from Example~\ref{ex:bridge} where a strict local martingale with the property that $S_T = 1 - \mm_T$, almost surely is given.
As a consequence of Theorem~\ref{thm:ExistenceIV} below, it also follows that for each given strict local martingale $S$ with defect $\mm_T$, 
there exists $x^*\in [\log(\mm_T), 0]$ such that the corresponding Call price $C_S(x)$ violates 
the no static-arbitrage lower bound $(1-\E^{x})_+$ for all $x < x^*$. 

These pathologies appear to contradict the fact that strict local martingale models are consistent with absence of arbitrage (in the NFLVR sense). 
Consider for example a Call option that is valued at the lower bound in~\eqref{eq:call_bounds}, that is $C_S(x) = (1 - \mm_T - \E^x)_+$. 
Then, choosing a log-strike $x$ such that  $x \le \log(1 - \mm_T)$, it is possible to form a costless portfolio consisting of a long position in the Call $C_S(x)$, a short position of one unit of the stock and $\mm_T + \E^x$ in the bank account. 
The payoff at maturity of this portfolio is $\mm_T + (\E^{x}- S_T)_+ > 0$, 
and we have apparently constructed an arbitrage. The resolution of this paradox is that due to the short position in $S$ the value process of the portfolio is unbounded from below and therefore not an admissible strategy in the sense of~\cite[Definition 2.7]{DelScha94}. 
Indeed, if the portfolio value were bounded from below then $S$ would be bounded from above and hence a true martingale, in contradiction to the strict martingale property of $S$. 

These considerations regarding admissibility of strategies are not entirely academic. 
Short positions usually require deposit of collateral, which restricts the scope of implementable hedging strategies. 
For this reason it has been remarked by several authors~\cite{CoxHob, HestonBubbles, MadanYor} that 
defining the Call-price by the risk-neutral expectation of its payoff is not the only economically sensible choice. 
We follow here the definition in~\cite{CoxHob}, where the authors consider a short position in Calls and argue that such positions are usually subject to collateral requirements. Thus, the value process $V$ of a hedging portfolio must not only replicate the Call payoff at maturity $T$, but also satisfy the collateral requirement $V_t \ge G(S_t)$ at intermediate times. 
Here, the function $G$ is used to describe the amount of collateral needed in relation to the stock price. The following is proved in~\cite{CoxHob}:
\begin{theorem}[Theorem 5.2 in~\cite{CoxHob}]\label{thm:Collateral}
Let~$G$ be a positive convex function satisfying $\limsup_{s\uparrow\infty}\frac{G(s)}{s} = \alpha$,
and~$H$ an arbitrary payoff satisfying $H\geq G$, then the fair price (at inception) of a European option with payoff~$H(S_T)$ is equal to 
$\EE^{\QQ}(H(S_T)) + \alpha \mm_T$.
\end{theorem}
By `fair price', Cox and Hobson mean the smallest initial fortune required to construct a self-financing wealth process super-replicating both the payoff at maturity and the collateral requirement along the life of the contract. 
In the case of a European Call option (with maturity $T$ and strike $\E^{x}$),
$H(S_T) = (S_T - \E^x)_+$, 
Theorem~\ref{thm:Collateral} implies that the fair price of the Call under collateralisation is given by 
\begin{equation}\label{eq:collateralized_call}
C_S^\alpha(x) := C_S(x) + \alpha \mm_T.
\end{equation}
We shall refer to $C_S^\alpha(x)$ as the value of an $\alpha$-collateralised Call. 
Note that only values $\alpha \in [0,1]$ make sense due to the requirement on $G$. 
It is clear that the collateral requirement does not affect the Call price in a true martingale model where $\mm_T = 0$. 
In a strict local martingale model prices of collateralised Calls differ from prices of uncollateralised Calls. 
Of particular interest is the fully collateralised Call $C_S^1(x)$, 
which coincides with the European Call price proposed by Madan and Yor in~\cite{MadanYor} for 
strict local martingale models: 
let~$(\tau_n)_{n\geq 1}$ be any sequence of stopping times increasing to infinity,
such that the stopped process $(S_{t\wedge\tau_n})_{t\geq 0}$ is a uniformly integrable martingale
for each $n\geq 1$.
The Madan-Yor European option price is then defined as
$C^{\mathrm{MY}}_S(x):=\lim_{n\uparrow\infty}\EE^{\QQ}\left(S_{T\wedge\tau_n} - \E^{x}\right)_+$.
\begin{proposition}[Proposition 2 in~\cite{MadanYor}]\label{prop:MadanYor}
For any $x\in\RR$, the following equalities hold:
$$
C^{\mathrm{MY}}_S(x)
 = (1-\E^{x})_+ + \frac{1}{2}\EE^{\QQ}(\Ll_T^x)
 = C_S(x) + \pi_T^S,
$$
where $(\Ll_t^x)_{t\geq 0}$ denotes the local time of~$S$ at level~$\E^{x}$,
and where the correction term reads
$$
\pi_T^S
 = \lim_{z\uparrow\infty}z\QQ\left(\sup_{0\leq u \leq T}S_u \geq z\right).
$$
\end{proposition}

A non-zero correction term $\pi_T^s$ corresponds to the necessary and sufficient condition for a price bubble given in~\cite[Theorem 3.4, Equation (5)]{CoxHob}. In~\cite[Appendix]{CoxHob} it is also shown that in fact $\pi_T^s$ is equal to the martingale defect $\mm_T$ and thus the Madan-Yor Call price $C^{\mathrm{MY}}_S(x)$ coincides with the fully collateralised Call price $C^1_S(x)$. 
Also the Call prices $G^2$ and $G^1$ discussed in~\cite{HestonBubbles} in the context of strict local martingale models correspond precisely to the uncollateralised and the fully collateralised Call price $C_S=C_S^0$ and $C_S^1$ respectively. 
Finally, also the `generalised fair value' of the Call discussed in~\cite[Chapter~5]{Lewis} is exactly the fully collateralised Call price. 
Inserting into~\eqref{eq:pc_parity} and~\eqref{eq:call_bounds} we see that
\begin{equation}\label{eq:pc_parity_full}
C^1_S(x) - P_S(x) = 1 - \E^x,
\qquad\text{and}\qquad
\min(\mm_T,1 - \E^x) \le C^1_S(x) \le 1,
\end{equation}
that is for the fully collateralised Call price, Put-Call-parity is restored and the pricing bounds are always within the no-static-arbitrage region.

Finally, let us calculate the following limits for large strikes, which will be needed later. 
\begin{equation}\label{eq:large_strike}
\begin{split}
\lim_{x\uparrow\infty}C^\alpha_S(x)
 & = \lim_{x\uparrow\infty} C_S(x) + \alpha \mm_T = \lim_{x \uparrow \infty} \EE^{\QQ}(S_T - \E^x)_+ + \alpha \mm_T = \alpha \mm_T\\
\lim_{x\uparrow\infty}\left[P_S(x) - \E^{x}\right]
 & = \lim_{x\uparrow\infty}\EE^{\QQ}\left((\E^x - S_T)_+ - \E^x\right)
 = - \lim_{x\uparrow\infty}\EE^{\QQ}\min(\E^x,S_T) = \mm_T - 1,
\end{split}
\end{equation}
where we have used dominated convergence and the supermartingale property of $S$ to evaluate the limits. Since~$\mm_T$ appears explicitly, the large-strike limit of Put and Call prices can be used to characterise the strict local martingale property of $S$, see also~\cite[Theorem~3.4(iv)]{CoxHob}.

\subsection{Duality with respect to true martingales with mass at zero}\label{sub:duality}
An important concept related to positive strict local martingales, is their duality relationship to true martingales with mass at zero.

\begin{definition}\label{def:duality}
Let $\QQ$ and $\PP$ be probability measures on a filtered measure space and let $T > 0$ be a fixed time horizon. 
Let $S$ be a strictly positive local $\QQ$-martingale and $M$ be a non-negative true $\PP$-martingale on $[0,T]$. Denote by $\tau:=\inf\{t>0: M_t=0\}$ the first hitting time (of~$M$) of zero and assume that $\tau$ is predictable and $\tau > 0$, $\PP$-a.s. 
We say that the pair~$(S,\QQ)$ is in duality to $(M,\PP)$ if $\QQ$ is absolutely continuous with respect to~$\PP$ on~$\Ff_T$, with 
\[\left.\frac{\D\QQ}{\D\PP}\right|_{\Ff_T} = M_T \qquad \text{and} \qquad S_t = \frac{1}{M_t} \quad \text{$\PP$-a.s. on $\{t < \tau \wedge T\}$}.\]
\end{definition}
Note that the above definition requires that $S$ is a strictly positive local martingale, 
a slightly stronger assumption than the non-negativity assumption made in Section~\ref{sub:market}.  
In financial modelling, $\QQ$ can be interpreted as the `share measure' corresponding to the stock price~$M$ 
under~$\PP$ or---in the context of currency models---as the `foreign measure' corresponding to the domestic measure~$\PP$ and the exchange rate process~$M$~\cite[Chapter~17]{Bjork}. 
Also note that since~$M$ is a true non-negative~$\PP$-martingale, 
zero is necessarily absorbing~\cite[Chapter III, Lemma 3.6]{Jacod}, and hence $M_t = 0$ for all $t \ge \tau$.
For models in duality,
\begin{equation}\label{eq:mm_dual}
\mm_t = 1 - \EE^{\QQ}(S_t) = 1 - \EE^{\PP}(\ind_{\{t< \tau\}}) = \PP(\tau\leq t) = \PP(M_t = 0),
\end{equation}
that is the martingale defect of $S$ (under $\QQ$) equals the mass at zero of $M$ (under $\PP)$.\\

The following result is deep, but by now well-understood. A proof can be found e.g. in~\cite{KardNik}; see also~\cite{Pal, Ruf} for similar versions.
\begin{lemma}
Let $(\Ff_t)_{t \geq 0}$ be right continuous and a standard system and let $T > 0$ be a fixed time horizon. 
For any pair $(S,\QQ)$ satisfying the assumptions of Definition~\ref{def:duality} there exists a dual pair $(M,\PP)$. Conversely, for any pair $(M,\PP)$ and associated stopping time $\tau$ satisfying the assumptions of Definition~\ref{def:duality}, there exists a dual pair $(S,\QQ)$. 
\end{lemma}
Let us remark that the difficult direction is going from $(S,\QQ)$ to $(M,\PP)$ in the case where $S$ is a strict local martingale; the required construction of $\PP$ relies on the F\"ollmer exit measure, first introduced in~\cite{Follmer} and mentioned in~\cite{Meyer}. Going from $(M,\PP)$ to $(S,\QQ)$ is easier, and an early proof under the assumption of continuity can be found in Delbaen and Schachermayer~\cite{DelScha}. 
For our purposes, the technical condition of $(\Ff_t)_{t\geq 0}$ being right continuous and a standard system can be satisfied by taking this filtration to be the right-continuous modification of the natural filtration of the coordinate process on the Skorokhod space $D(\RR_{\geq 0},\RR \cup \{+\infty\})$, i.e., the space of c\`adl\`ag functions with possible explosion to infinity in finite time. 
See discussions in~\cite{KardNik} and~\cite{Follmer} for details.
For models in duality we have the following simple characterisation of the strict local marginale property of $S$ under $\QQ$.

\begin{lemma}\label{lem:duality_mass}
Let $(M,\PP)$ and $(S,\QQ)$ be market models in duality with time horizon $T > 0$. The following are equivalent:
\begin{enumerate}[(i)]
\item $S$ is a strict local $\QQ$-martingale on $[0,T]$;
\item $M_T$ has mass at zero, i.e., $\PP(M_T = 0) > 0$;
\item $\mm_T > 0$;
\item $\QQ$ is not equivalent to $\PP$ on $\Ff_T$.
\end{enumerate}
\end{lemma}
\begin{proof}
From Definition~\ref{def:duality} it is clear that $S$ is a local martingale, hence the strict local martingale property of $S$ is equivalent to $\mm_T = 1 - \EE^{\QQ}(S_T) > 0$. 
In view of Equation~\eqref{eq:mm_dual} this shows equivalence of the first three assertions. 
Finally $\left.\frac{\D\QQ}{\D\PP}\right|_{\Ff_T} = M_T$ and hence~$\QQ$ is not equivalent to~$\PP$ 
if and only if $M_T = 0$ with positive $\PP$-probability, which is exactly assertion~(ii). 
\end{proof}

\begin{example}\label{ex:diffusion2}
We now continue Example~\ref{ex:diffusion}, i.e., we consider $M$ given by 
$\D M_t = \sigma(M_t)\D W_t$, where~$\sigma$ satisfies the same conditions as before. 
In addition assume $\int_1^\infty x \sigma^{-2}(x) \D x < \infty$, 
such that $M$ is a true $\PP$-martingale by~\eqref{eq:integral2}, and let $\tau$ denote the first hitting time of zero of $M$. 
Construct now the measure~$\QQ$ via $\left.\frac{\D\QQ}{\D\PP}\right|_{\Ff_t} = M_t$, 
and set $S_t:=M_{t}^{-1}\ind_{\{t<\tau\}}$. It is easy to see that $S$
satisfies the stochastic differential equation
$\D S_t = \widetilde{\sigma}(S_t)\D W_t^{\QQ}$,
with $S_0=1$, $\widetilde{\sigma}(y) =  y^2\sigma(1/y)$, 
and where $W^{\QQ}$ is a $\QQ$-Brownian motion defined up to time $\tau$.
Note that the following equalities hold:
\begin{equation}\label{eq:example_diffusion}
\int_{0}^{1}\frac{y}{\widetilde{\sigma}^2(y)}\D y = \int_{1}^{\infty}\frac{x}{\sigma^2(x)}\D x
\qquad\text{and}\qquad
\int_{1}^{\infty}\frac{y}{\widetilde{\sigma}^2(y)}\D y = \int_{0}^{1}\frac{x}{\sigma^2(x)}\D x.
\end{equation}
The first integral is finite by assumption, and it follows from the integral conditions~\eqref{eq:integral1} and~\eqref{eq:integral2} in Example~\ref{ex:diffusion} that $S$ is $\QQ$-a.s. strictly positive. 
If the second integral is finite it follows from the same conditions that also~$M$ is $\PP$-a.s. positive and~$S$ is a true $\QQ$-martingale. 
If the second integral is infinite, then $M$ has mass at zero and $S$ is a strict local martingale. Other cases are not possible, in line with Lemma~\ref{lem:duality_mass}.
\end{example}

Relations between option prices in dual models are well known in the case where 
both~$S$ and~$M$ are true martingales without mass at zero~\cite{Eberlein, Geman})
and have been described in the strict local martingale case in~\cite{KardNik, Ruf}. 
We recall the relation for European Puts and Calls. 
In addition to the Put and Call prices $P_S$ and $C^\alpha_S$ described above, we write 
$P_M(x) := \EE^{\PP}( \E^x - M_T)_+$ and $C_M(x) := \EE^{\PP}(M_T - \E^x)_+$ for the European Put and Call option on the underlying~$M$ under $\PP$.
\begin{proposition}\label{prop:PricesStrictLocal}
Let $(M,\PP)$ and $(S,\QQ)$ be market models in duality with time horizon $T > 0$. 
Then for any $x\in\RR$ and any $\alpha\in [0,1]$, the following relations between Put and Call prices hold:
\begin{equation}\label{eq:CallPutStrict}
C_S^\alpha(x) = \E^{x}P_M(-x) + (\alpha-1) \mm_T
\qquad\text{and}\qquad
P_S(x) = \E^{x}C_M(-x).
\end{equation}
\end{proposition}
\begin{proof}
For Puts we compute
\begin{align*}
P_S(x)
 & = \EE^{\QQ}(\E^{x} - S_T)_+
  = \EE^{\PP}\left[M_T^\tau \left(\E^{x} - S_T\right)_+\right]
 = \EE^{\PP}\left[M_T^\tau \left(\E^{x} - \frac{1}{M_T}\ind_{\{T<\tau\}}\right)_+\right]\\
 & = \E^{x}\EE^{\PP}\left[\left(\ind_{\{T<\tau\}}+\ind_{\{T\geq\tau\}}\right)\left(M_T^\tau - 
 \E^{-x}\ind_{\{T<\tau\}}\right)_+\right]\\
 & = \E^{x}\EE^{\PP}\left[\ind_{\{T<\tau\}}\left(M_{T} - \E^{-x}\right)_+\right]
  = \E^{x}\EE^{\PP}\left[\left(1 - \ind_{\{T\geq\tau\}}\right)
 \left(M_{T} - \E^{-x}\right)_+\right]
 = \E^{x}C_M(-x).
\end{align*}
The results for Calls follows from Put-Call parity for the $\PP$-martingale~$M$ and `modified Put-Call parity'~\eqref{eq:pc_parity} for the local $\QQ$-martingale~$S$.
\end{proof}


\section{Implied volatility for strict local martingales}
For each $x\in\RR$, the implied volatility for a given Call price~$C(x)$ 
is defined as the unique non-negative solution to the equation $C_{\BS}(x,\sigma) = C(x)$, 
where $C_{\BS}$ represents the Black-Scholes European Call price with maturity~$T$, strike~$\E^{x}$
and volatility~$\sigma$:
$$
C_{\BS}(x,\sigma) := \Nn\left(d_+(x,\sigma)\right)
 - \E^{x} \Nn\left(d_-(x,\sigma)\right),
\qquad\text{where }
d_{\pm}(x,\sigma) := \frac{-x}{\sigma\sqrt{T}} \pm \frac{1}{2}\sigma\sqrt{T},
$$
with $\Nn$ standing for the Gaussian cumulative distribution function.
It is known that the implied volatility is a well-defined real number in $[0,\infty)$ if and only if~$C(x)$ 
lies within the no-static-arbitrage bounds (given by the bounds~\eqref{eq:call_bounds} for $\mm_T = 0$) 
and that in this case it is unique. 
If Put-Call-parity holds then the definition using European Put options is equivalent to that using Call options . 
Hence in true martingale models, the implied volatility is always uniquely defined and there is no distinction between Call- and Put-implied volatility. 
The behaviour of the implied volatility in such models is by now fairly well understood, see e.g.~\cite{GatheralBook}. However, in the strict local martingale case, surprisingly few results exist
(apart from~\cite{Tehranchi}, which studies the large-time behaviour). 
It turns out that in the strict local martingale setting, 
even existence of implied volatilities is not certain and one cannot equivalently consider Call and Put options.
In this section, we endeavour to fill this gap by providing results on existence, uniqueness and on the asymptotic behaviour for large strikes of implied volatilities in the class of strict local martingale models.


\subsection{Put-and Call-implied volatility}
As discussed above, Put-Call parity fails in the strict local martingale setting (unless Calls are fully collateralised), 
and hence Call-implied volatilities have to be distinguished from Put-implied volatilities. 
We denote by $I^p_S(x)$ the implied volatility corresponding to the price $P_S(x)$ of a Put with log-strike $x$, 
written on a local martingale~$S$, and denote by~$I^\alpha_S(x)$ the implied volatility corresponding to the $\alpha$-collateralised Call~$C^\alpha_S(x)$, 
namely, for each $x\in\RR$, the unique non-negative solution (whenever it exists) to the equation
$C_{\BS}(x, I_S^\alpha(x)) = C_S^\alpha(x)$.
We start with the following result discussing the existence of~$I^p_S$ and $I^\alpha_S$.
\begin{theorem}\label{thm:ExistenceIV}
Let~$S$ be a non-negative local martingale. 
\begin{enumerate}[(i)]
\item The implied volatility $I^p_S$ of the Put $P_S$ is well defined on the whole real line;
\item The implied volatility $I^1_S$ of the fully collateralised Call $C^1_S$ is well defined 
on~$\RR$ and coincides with the Put-implied volatility: $I^1_S(x) = I^p_S(x)$,
for all $x \in \RR$;
\item For $\alpha \in [0,1)$ there exists~$x^*(\alpha)\leq0$ such that the implied volatility~$I^\alpha_S$ of the $\alpha$-collateralised Call is well defined on~$[x^*(\alpha),+\infty)$, 
but not on~$(-\infty, x^*(\alpha))$. 
The function $x_*(\alpha)$ is strictly decreasing and satisfies
\begin{equation}\label{eq:alpha_bounds}
\log\left((1-\alpha\right)\mm_T) < x^*(\alpha) \le \log(1 - \alpha \mm_T).
\end{equation}
For every $x \in \RR$ the function $\alpha \mapsto I^\alpha_S(x)$ is strictly increasing on the interval where it is defined, and $I_S^\alpha(x) < I_S^p(x)$ holds for all $\alpha \in [0,1)$ 
and $x \in [x^*(\alpha),+\infty)$.
\end{enumerate}
\end{theorem}
\begin{remark}
Reparameterising by log-strike and setting $K^*(\alpha) = \exp(x^*(\alpha))$ the bounds in~\eqref{eq:alpha_bounds} simplify to $(1 - \alpha) \mm_T < K^*(\alpha) \le (1 - \alpha \mm_T)$. 
Even without specifying a concrete model for $S$, 
the region $\Dd := \{(\alpha,x): \alpha \in [0,1), x \in \RR\}$
can be written as the disjoint union
$\Dd = \Dd^A \cup \Dd^N \cup \Dd^M$, where
\begin{align*}
\Dd^A &:= \{(\alpha,x) \in \Dd: x > \log(1 - \alpha \mm_T)\},\\
\Dd^N &:= \{(\alpha,x) \in \Dd: x \le \log((1 - \alpha) \mm_T)\},\\
\Dd^M &:= \{(\alpha,x) \in \Dd: \log((1 - \alpha) \mm_T) < x \le \log((1 - \alpha) \mm_T)\},
\end{align*}
such that the implied volatility $I_S^\alpha(x)$ is \emph{a}lways defined in $\Dd^A$, 
\emph{n}ever in $\Dd^N$ and \emph{m}ay or may not be in $\Dd^M$. It will also become clear from the proof that the region 
where $I_S^\alpha(x)$ is not defined is precisely the region where the Call price $C_S^\alpha(x)$ violates the lower no-static-arbitrage bound. 
\end{remark}

\begin{proof}
The price of the Put and the price of the fully collateralised Call are always inside the no-static-arbitrage region by~\eqref{eq:call_bounds} and~\eqref{eq:put_bounds}. 
Hence, the corresponding implied volatilities $I_S^p(x)$ and $I_S^1(x)$ are well defined for all $x \in \RR$. 
The Put-Call parity~\eqref{eq:pc_parity_full} for the fully-collateralised Call and for the Black-Scholes price yields
$$
P_S(x) = C_S^1(x) + \E^x - 1 = C_\text{BS}(x,I_S^1(x)) + \E^x - 1= P_\text{BS}(x,I_S^1(x)).
$$
Since $P_S(x) = P_\text{BS}(x,I_S^p(x))$, uniqueness of implied volatility implies that 
$I_S^1(x) = I_S^p(x)$ for all $x \in \RR$,
and Claims~(i)~and~(ii) follow.\\
Let now $\alpha \in [0,1)$. The implied volatility $I^\alpha(x)$  for the Call $C_S^\alpha(x)$ exists if and only if $C_S^\alpha(x)$ is inside the no-static-arbitrage region $[(1-\E^x)_+,1)$. 
From~\eqref{eq:call_bounds} we derive the bounds 
\begin{equation}\label{eq:alpha_call_bounds}
(1 - \mm_T - \E^x)_+  + \alpha \mm_T \le C_S^\alpha(x) < 1 + (\alpha - 1)\mm_T.
\end{equation}
Thus, $C_S^\alpha(x)$ is in the no-static-arbitrage region if and only if $C_S^\alpha(x) \ge (1 - \E^x)_+$ or equivalently if $F^\alpha(x) := C_S^\alpha(x) - (1 - \E^x)_+ \ge 0$. 
To find the zeros of $F^\alpha(x)$ is suffices to consider $x < 0$ since $F^\alpha(x) > 0$ for all $x \ge 0$. Rewriting $F^\alpha$ for $x \le 0$ as 
\[
F^\alpha(x) = \EE^{\QQ}(S_T - \E^x)_+ + \alpha \mm_T - (1 - \E^x) = \EE^{\QQ}\left(\max(S_T,\E^x)\right)  +\alpha \mm_T - 1,
\]
we see that $F^\alpha$ is continuous and increasing on $(-\infty,0]$ with 
$\lim_{x \downarrow -\infty} F^\alpha(x) = \mm_T(1 - \alpha)$. 
Setting $x^*(\alpha) := \inf \{x\leq 0: F^\alpha(x) \ge 0\}$, it follows that implied volatility exists on $[x^*(\alpha),\infty)$ but not on $(-\infty,x_*(\alpha))$. 
It is also clear that for fixed $x$ the function $\alpha \mapsto F^\alpha(x)$ is strictly increasing, and hence that $x^*(\alpha)$ must be strictly decreasing. 
Moreover, considering the left limit of $F^\alpha(x)$ at $-\infty$ it follows that $x^*(\alpha) = -\infty$ for $\alpha = 1$ and $x^*(\alpha) > -\infty$ for all other $\alpha \in [0,1)$. 
In the latter case it holds that $C_S^\alpha(x^*(\alpha)) = 1- \E^{x^*(\alpha)}$. 
Plugging the right-hand side into the bounds~\eqref{eq:alpha_call_bounds} and rearranging we obtain~\eqref{eq:alpha_bounds}.
\end{proof}



\subsection{Asymptotic behaviour of the implied volatility}

For large strikes, the following result provides the asymptotic behaviour of Put-and Call-implied volatilities.
\begin{theorem}\label{thm:IVExpansion}
Let $S$ be a non-negative strict local martingale with martingale defect $\mm_T$ 
and suppose that $\alpha > 0$. Then, as~$x$ tends to infinity, the following expansions hold:
$$
I^p_S(x) = I^1_S(x) = \sqrt{\frac{2x}{T}} + \frac{\Nn^{-1}(\mm_T)}{\sqrt{T}} + o(1) \qquad\text{and}\qquad
I^\alpha_S(x) = \sqrt{\frac{2x}{T}} + \frac{\Nn^{-1}(\alpha\mm_T)}{\sqrt{T}} + o(1).
$$
\end{theorem}
\begin{remark}
The $o(1)$ term can be made more precise with some further assumptions on the right tail of the distribution of $S_T$ along the lines of~\cite{DMHJ, Guli}. 
In fact, as we shall show in Section~\ref{sub:dualIV} under the mild additional assumption that $S$ is strictly positive 
the results of~\cite{DMHJ, Guli} can be directly translated into higher-order expansions of $I^p_S(x)$ and $I^\alpha_S(x)$ using the duality approach of Section~\ref{sub:duality}
\end{remark}
For the implied volatility of uncollateralised Calls and in true martingale models the following complementary result holds:
\begin{corollary}\label{cor:IV}
If $\alpha = 0$, then 
$$
\lim_{x \uparrow \infty} \left(I^0_S(x)  - \sqrt{\frac{2x}{T}} \right)  = -\infty.
$$
If $\mm_T = 0$, then, for all $\alpha \in [0,1]$,
$$
\lim_{x \uparrow \infty} \left(I^p_S(x)  - \sqrt{\frac{2x}{T}} \right)
 = \lim_{x \uparrow \infty} \left(I^\alpha_S(x)  - \sqrt{\frac{2x}{T}} \right)  = -\infty.
$$
\end{corollary}
\begin{remark}
Together, Theorem~\ref{thm:IVExpansion} and Corollary~\ref{cor:IV} show that there is a sharp distinction between the behaviour of the implied volatility for large strikes in strict local martingale models and in true martingale models. This observation can be used to detect the strict local martingale property from observed implied volatilities, see Section~\ref{sub:test} below.
\end{remark}

\begin{proof}
By Theorem~\ref{thm:ExistenceIV}, $I^p_S(x) = I^1_S(x)$, so that it suffices to consider the Call-implied volatility $I^\alpha_S(x)$ for $\alpha \in [0,1]$. 
By the same result, $I^\alpha_S(x)$ is well-defined at least for all $x \in (0,\infty)$. By~\eqref{eq:large_strike}, $I^\alpha_S(x)$ must satisfy
\[
 \lim_{x \uparrow \infty} \Nn\Big(d_+(x,I^\alpha_S(x))\Big) - \E^x \Nn\Big(d_-(x,I^\alpha_S(x))\Big)  = 
\lim_{x \uparrow \infty} C_\text{BS}(x,I^\alpha_S(x)) = \lim_{x \uparrow \infty} C^\alpha_S(x) =  \alpha \mm_T.
\]
The arithmetic-geometric-mean inequality yields $d_-(x,\sigma) \le - \sqrt{2x}$, and hence
\[
\lim_{x \uparrow \infty} \E^x \Nn\Big(d_-(x,I^\alpha_S(x))\Big) \le \lim_{x \uparrow \infty} \E^x \Nn\Big(- \sqrt{2x}\Big) \le \lim_{x \uparrow \infty} \E^x \frac{\phi(\sqrt{2x})}{\sqrt{2x}} = 0,
\]
where we have used the classical bound $\Nn(-x) / \phi(x) \le x^{-1}$ on Mills ratio. Hence 
\[
 \lim_{x \uparrow \infty} \Nn\Big(d_+(x,I^\alpha_S(x))\Big) = \lim_{x \uparrow \infty} C^\alpha_S(x) =  \alpha \mm_T,
\]
and inverting $\Nn$ we obtain 
\begin{equation}\label{eq:d1_limit}
\lim_{x \uparrow \infty} d_+(x,I^\alpha_S(x)) = \Nn^{-1}(\alpha \mm_T),
\end{equation}
where we set $\Nn^{-1}(0) = -\infty$ and $\Nn^{-1}(1) = +\infty$.

Suppose now, that for $x > 0$ and $\sigma \ge 0$ we have lower and upper bounds
$$
l \le d_+(x,\sigma) \le u,
$$
for some $l,u \in \RR$. 
Using the explicit form of $d_+$ and solving the quadratic equation, this implies that
\begin{equation}\label{eq:d1_bounds}
l + \sqrt{l^2 + 2x} \le \sigma \sqrt{T} \le u  + \sqrt{u^2 + 2x}.
\end{equation}
If $\alpha \mm_T > 0$, then combining this estimate with~\eqref{eq:d1_limit} and expanding the square root, we obtain
\[I^\alpha_S(x) = \sqrt{\frac{2x}{T}} + \frac{\Nn^{-1}(\alpha\mm)}{\sqrt{T}} + o(1)\]
for large~$x$ and Theorem~\ref{thm:IVExpansion} follows. 
For Corollary~\ref{cor:IV} we have to consider the degenerate case $\alpha \mm_T = 0$, or equivalently $\Nn^{-1}(\alpha \mm_T) = -\infty$. 
In this case, we can still apply the upper bound in~\eqref{eq:d1_bounds}, which holds for arbitrary $u \in \RR$. Hence, 
\[\lim_{x \uparrow \infty} \left(I^\alpha_S(x)\sqrt{T}  - \sqrt{2x} \right) \le \lim_{x \uparrow \infty} \Big(-u + \sqrt{2x + u^2} - \sqrt{2x}\Big)  = -u \]
for arbitrary $u \in \RR$, and Corollary~\ref{cor:IV} follows.
\end{proof}


\subsection{Relation to Lee's~\cite{Lee} and Benaim-Friz'~\cite{BenaimFriz} asymptotics}

Roger Lee~\cite{Lee} pioneered the analysis of the tail behaviour of the implied volatility under 
the true martingale framework.
He proved that, for a given strictly positive $\PP$-martingale~$X$, 
the large-strike behaviour of the implied volatility $I_X$ is given by 
\begin{equation}\label{eq:LeeSlope}
\limsup_{x\uparrow\infty}\frac{I_X(x)^2 T}{x} = \psi(p^*) \in [0,2],
\end{equation}
where $p^*:=\sup\{p\geq 0: \EE^{\PP}(X_T)<\infty\}$ and $\psi(p)\equiv 2-4(\sqrt{p(p+1)}-p)$.
Similar results also hold for the small-strike behaviour using the negative moments of~$X_T$.
Subsequently, Benaim and Friz~\cite{BenaimFriz} refined~\eqref{eq:LeeSlope}
by providing sufficient conditions under which the $\limsup$ can be strengthened into a genuine limit.
Surprisingly though, the martingale assumption is not explicitly mentioned in~\cite{BenaimFriz}, 
though it is in the review paper~\cite{BenaimFrizLee}.
An immediate consequence of~\eqref{eq:LeeSlope} (or its sharpened version in~\cite{BenaimFriz}) is that 
when $X$ is a $\PP$-martingale, the large-strike slope 
of the total implied variance~$I_X^2(x)T/x$ is equal to two if and only if $p^*=1$, 
namely no moment strictly greater than one exists, so that the distribution of~$X_T$ has a fat right tail.

In the strict local martingale framework this one-to-one correspondence between tail-weight and slope of implied volatility breaks down. 
Consider for instance the strict local martingale $\D S_t = S_t^2 \D W_t^{\QQ}$, starting at $S_0 = 1$,
whose density is given in~\eqref{eq:densityCEV2} in the appendix. A simple Taylor expansion shows that, as $s$ tends to infinity, the following asymptotic behaviour holds
for any $p\geq 0$:
$$
s^p \PP(S_T \in \D s) = \sqrt{\frac{2}{\pi T^3}}\E^{-1/(2T)} s^{p-4}
\left\{1 + \mathcal{O}(s^{-2})\right\}\D s,
$$
and $p^*_S:=\sup\{p\geq 0: \EE^{\PP}(S_T)<\infty\}=3$. 
Should Lee's (or Benaim-Friz') formula hold, then
$\limsup_{x\uparrow\infty}I_s^\alpha(x)^2T/x = \psi(p^*_S)<2$, with $\alpha \in (0,1]$. 
This stands in contradiction to Theorem~\ref{thm:IVExpansion}, which states that 
$\limsup_{x\uparrow\infty}\frac{I_s^\alpha(x)^2T}{x} = 2$. 
This example shows that the results of Lee's~\cite{Lee} and Benaim-Friz'~\cite{BenaimFriz} 
cannot hold for strict local martingales.

\subsection{Testable implications of price bubbles and $\eps$-close implied volatility}\label{sub:test}
As discussed at the beginning of this paper, strict local martingale models have been advocated as models for stock price bubbles. 
In this regard, it is of interest to be able to test empirical data for the appearance of such bubbles. 
In the context of continuous Markov diffusions, such tests have been proposed and implemented by~\cite{Jarrow} 
(see also~\cite{Hulley} for a similar idea) based on statistical estimation of historical volatility. 
Once the diffusion coefficient of the stock price is estimated and extrapolated to the whole real half-line, 
the integral criterion discussed in Example~\ref{ex:diffusion} is used by~\cite{Jarrow} to decide whether the underlying is a strict local martingale or not. Complementary to the statistical approach of~\cite{Jarrow}, our results suggest different ways to test for the appearance of a stock price bubble 
based on implied (as opposed to historical) volatility. 
First, observe that in the case of non-fully-collateralised Calls ($\alpha < 1$) there are simple criteria to distinguish between true and strictly local martingales, based on implied volatility. 
From the results presented above, it follows that unless Calls are fully collateralised there is equivalence of the following statements:
\begin{itemize}
\item The stock price process~$S$ is a strict local martingale under the pricing measure~$\QQ$;
\item Put-and Call-implied volatilities are different;
\item The Call-implied volatility does not exist for sufficiently small strikes.
\end{itemize}
However, for fully collateralised Calls ($\alpha = 1$) these criteria fail. Instead, we derive from Theorem~\ref{thm:IVExpansion} the following criterion:
\begin{itemize}
\item The implied volatility satisfies $I_S(x) = \sqrt{\frac{2x}{T}} + \frac{\nn_T}{\sqrt{T}} + o(1)$ for some $\nn_T \in \RR$.
\end{itemize}
Note that the last criterion is necessary and sufficient.
The necessary part comes from Theorem~\ref{thm:IVExpansion},
and it is sufficient by Corollary~\ref{cor:IV} applied to the case $\mm_T=0$.
In addition, the martingale defect $\mm_T$ can be extracted by setting $\mm_T = \Nn(\nn_T)$.
The drawback of this criterion is that it is an \emph{asymptotic} test, valid only for large $x$. 
This drawback is shared with the statistical test of~\cite{Jarrow} which also requires asymptotic extrapolation of the estimated diffusion coefficient. This property limits the value of the test in practical applications, since implied volatility (or Call prices) can only be observed at a finite number of strikes. 
However, a simple argument shows that any test to determine the strict local martingale property from Put-implied (or fully collateralised Call-implied) volatilities is necessarily an asymptotic test in the sense that it uses arbitrarily large strike values as input: 
let $S$ be a non-negative local $\QQ$-martingale with localising sequence $(\tau_n)_{n \in \mathbb{N}}$, 
so that the stopped processes $S_t^n := S_{t \wedge \tau_n}$ are true $\QQ$-martingales. 
The difference between Put prices on $S$ and $S^n$  can be estimated uniformly on time-strike rectangles $\mathcal{R}_{T,\tilde{x}} := [0,T] \times (-\infty,\tilde{x})$, for any $\tilde{x}\in\RR$, by
$$
\sup_{(t,x) \in \mathcal{R}_{T,\tilde{x}}} |P_S(x) - P_{S^n}(x)|
 = \sup_{(t,x) \in \mathcal{R}_{T,\tilde{x}}}  |\EE^\QQ(\E^x - S_T)_+ - \EE^\QQ(\E^x - S^n_T)_+|
 \le \E^{\tilde{x}} \QQ(\tau_n \ge T).
$$
The bound can be made arbitrarily small by choosing $n$ large enough. Since Put-implied volatilities depend continuously on the Put price this shows that for any non-negative strict local martingale $S$, time-strike rectangle $\mathcal{R}_{T,x_*}$ and $\varepsilon >0$ we can find a true martingale model $S^\varepsilon$, such that Put-implied volatilities (as well as those implied by fully collateralised Calls) are $\varepsilon$-close, uniformly on $\mathcal{R}_{T,x_*}$.

\subsection{Duality and Symmetry of the implied volatility}\label{sub:dualIV}

We consider the consequences of the duality relation studied in Section~\ref{sub:duality} on implied volatilities.

\begin{theorem}
Let~$S$ be a strictly positive strict local $\QQ$-martingale in duality with the true $\PP$-martingale~$M$ with mass at zero. 
Denote by $I_M(x)$ the implied volatility under~$\PP$ for log-strike~$x$ and underlying~$M$. 
Then, for all $x \in \RR$,
\begin{equation}
I^p_S(x) = I^1_S(x) = I_M(-x).
\end{equation}
\end{theorem}
\begin{proof}
For any $x\in\RR$, the implied volatility $I^p_S$ is the unique solution to the equation 
$P_{S}(x) = P_{\BS}(x, I^p_S(x))$.
Therefore, using~\eqref{eq:CallPutStrict}, we can write
\begin{align*}
\E^{x}\Nn(-d_-(x, I^p_S(x))) - \Nn(-d_+(x, I^p_S(-x)))
 & = P_{S}(x) = \E^{x}C_{\BS}(-x, I_M(-x)) \\
 & = \E^{x}\Nn(d_+(-x, I_M(-x))) - \Nn(d_-(-x, I_M(-x)))\\
 & = \E^{x}\Nn(-d_-(x, I_M(-x))) - \Nn(-d_+(x, I_M(-x))).
\end{align*}
The theorem follows by existence and uniqueness of the Put smile in the no-static-arbitrage region.
\end{proof}

Using the above theorem we can translate the results of~\cite{DMHJ} 
on the left-wing behaviour of the implied volatility in true martingale models with mass at zero 
directly into results on the right-wing behaviour of the implied volatility in strictly positive strict local martingale models:

\begin{corollary}\label{cor:dual}
Let~$S$ be a strictly positive strict local $\QQ$-martingale, $T > 0$ and~$\mm_T$ the martingale defect of~$S$. 
Set $G(x) := \EE^{\QQ}(S_T \ind_{\{S_T \ge \E^x\}})$ and $\nn_T := \Nn^{-1}(\mm_T)$.
\begin{enumerate}[(i)]
\item If $G(x) = o(x^{-1/2})$ as $x$ tends to infinity, then
\begin{equation*}
I_S^p(x)
 = I_S^1(x) = \sqrt{\frac{2x}{T}} + \frac{\nn_T}{\sqrt{T}} + \frac{\nn_T^2}{2\sqrt{2Tx}}
 + \frac{\exp(\frac{1}{2}\nn_T^2)}{\sqrt{2Tx}}\Psi(x),
\qquad \text{as $x$ tends to infinity},
\end{equation*}
where the function~$\Psi$ is such that $0 \le \limsup_{x \uparrow \infty} \Psi(x) \le 1$.
\item If $G(x) = \mathcal{O}(\E^{-\eps x})$ as $x$ tends to infinity, for some $\eps > 0$, then
\begin{equation*}
I_S^p(x) = I_S^1(x)
 = \sqrt{\frac{2x}{T}} + \frac{\nn_T}{\sqrt{T}} + \frac{\nn_T^2}{2 \sqrt{2Tx}}  + \Phi(x),
\qquad \text{as $x$ tends to infinity},
\end{equation*}
where the function~$\Phi$ satisfies $\limsup_{x \uparrow \infty} \sqrt{2Tx} |\Phi(x)| \le 1$.
\end{enumerate}
\end{corollary}
\begin{proof}
Let~$F$ denote the cumulative distribution function of $M_T$ under $\PP$. Then
\begin{equation*}
G(x) := \EE^{\QQ}\left(S_T \ind_{\{S_T \ge \E^{x}\}}\right)
 = \EE^{\PP}\left(\ind_{\{M_T \le \E^{-x}\}} \ind_{\{T < \tau\}}\right)
= F(\E^{-x}) - F(0).
\end{equation*}
The remaining claims follow from the duality relation $I_M(-x) = I^1_S(x)$  and~\cite[Theorem~1.1]{DMHJ}.
\end{proof}

In~\cite[Cor.~5.2]{DMHJ}, the authors showed that if the underlying is given by a true martingale with mass at zero, then the implied volatility smile cannot be symmetric (in the sense that $I_M(x) = I_M(-x)$). 
Applying the duality relation $I_M(-x) = I_S(x)$ we obtain the following.
\begin{corollary}\label{cor:no_symmetry}
If $S$ is a strictly positive strict local $\QQ$-martingale then the Put smile is not symmetric.
\end{corollary}
\begin{remark}
For Call-implied volatilities $I^\alpha_S$, the smile can never by symmetric for $\alpha \in [0,1)$, 
for the simple reason that $I^\alpha_S(x)$ does not exist for sufficiently small strikes, 
by Theorem~\ref{thm:ExistenceIV}. 
In the fully collateralised case $\alpha = 1$, the Call-implied volatility coincides with the Put-implied volatility and Corollary~\ref{cor:no_symmetry} applies.
\end{remark}

Again looking at the example of one-dimensional diffusions is instructive:

\begin{example}
We continue Examples~\ref{ex:diffusion} and~\ref{ex:diffusion2} and consider the $\PP$-martingale $M$ defined by 
$\D M_t = \sigma(M_t)\D W_t^\PP$ and the dual $\QQ$-martingale $S$ given by $\D S_t = \tilde \sigma(S_t) \D W_t^\QQ$, where $\tilde \sigma(y) = y^2\sigma(1/y)$. 
The dual models $(S,\QQ)$ and $(M,\PP)$ have the same distribution if $\sigma = \tilde \sigma$. 
In this case it follows from~\ref{eq:example_diffusion} and the discussion in Example~\ref{ex:diffusion2} 
that either $S$ and $M$ are both true martingales without mass at zero or they are both strict local martingales with mass at zero. 
Being a true martingale with mass at zero or a strict local martingale without mass at zero is not compatible with the symmetry assumption $\sigma = \tilde \sigma$, in line with~\cite[Cor.~5.2]{DMHJ} and Corollary~\ref{cor:no_symmetry}.
\end{example}

\subsection{Large-time behaviour}

In~\cite{Tehranchi}, Tehranchi studied the large-time behaviour of the implied volatility,
or more precisely of the total variance $I_S(x,T)^2 T$, where $I_S(x,T)$ now emphasises the dependency of implied volatility on maturity $T$.
His definition of the implied volatility, valid when the underlying stock price is a true martingale, is however not fully adequate for strict local martingales. 
We shall therefore understand it here as the Put-implied volatility~$I_S^p$, or equivalently,
from Theorem~\ref{thm:ExistenceIV}, as the fully collateralised Call implied volatility.
His main result reads as follows:
\begin{theorem}[Theorem 3.1 in~\cite{Tehranchi}]\label{thm:Tehranchi}
Let $(\Omega, \Ff, (\Ff_t)_{t\geq 0}, \QQ)$ be a given probability space, and~$S$ 
a non-negative $\QQ$-local martingale starting at one and such that $\QQ(S_t>0)>0$ for all $t\geq 0$.
If~$S_t$ converges almost surely to zero as~$t$ tends to infinity, then the following holds,
for any real number~$x$:
$$
I_S^p(x)^2 T = -8\log\EE^{\QQ}(S_T\wedge \E^x) - 4\log\left[-\log\EE^{\QQ}(S_T\wedge \E^{x})\right] + 4x-4\log(\pi) +\varepsilon(x,T),
$$
as~$T$ tends to infinity, where the error$~\varepsilon(\cdot, T)$ satisfies some uniform bounds on compacts
as~$T$ increases.
\end{theorem}
The assumption that~$S_t$ converges almost surely to zero is equivalent to Put option prices converging 
to the strike~$\E^{x}$ (by dominated convergence)--or the fully collateralised Call price converging to unity.
In~\cite[Lemma 3.3 ]{RogersTehranchi} it was shown that, if~$S$ is a $\QQ$-martingale, then this is also equivalent to 
$I_S(x,T)^2 T$ converging to infinity, for all $x\in\RR$.
An immediate computation shows that this still holds if~$S$ is a (non-negative) strict local $\QQ$-martingale.
Under some assumptions on the behaviour of the moment generating function of the stock price process, 
Tehranchi~\cite{Tehranchi} further expresses the right-hand side of Theorem~\ref{thm:Tehranchi}
in terms of the moment generating function. 
We shall not pursue this characterisation, however, since we are not making any assumptions on the moment generating function of the underlying stock price process.

Let now~$S$ be a non-negative $\QQ$-supermartingale. Since $S$ is bounded in $L^1(\QQ)$, the martingale convergence theorem~\cite[Theorem 69.1]{RogersWilliams} ensures that the limit 
$S_\infty := \lim_{t\uparrow\infty} S_t$ exists almost surely in~$[0,\infty)$.
It is however not immediately clear whether $S_\infty = 0$ almost surely for a general non-negative strict local $\QQ$-martingale, 
as required in Theorem~\ref{thm:Tehranchi}.
In the particular case where~$S$ has continuous paths on some interval $[0,\infty)$,
in natural scale and such that the origin is either natural or absorbing (see~\cite[Chapter 15, Section 6]{KarlinTaylor} 
for details of endpoints, but simply recall that this assumption essentially means that~$S$ is a strict local martingale), then Theorem 3.21 in~\cite{Hulley} guarantees that $S_\infty = 0$, $\QQ$-almost surely,
and therefore Theorem~\ref{thm:Tehranchi} holds.
Recently, Gushchin, Urusov and Zervos~\cite{GUZ14} proved that,
if $S$ is the weak solution to the stochastic differential equation
$\D S_t = \sigma(S_t) \D W_t$, with $S_0 = s_0 \in (0, \infty)$,
where $\sigma:[0,\infty) \to\RR$ is a Borel-measurable function satisfying $\sigma(s)=0$ if and only if $s=0$
and $\int_{A}\sigma^{-2}(s) \D s$ is finite for every non-empty finite interval~$A\subset (0,\infty)$, 
then~$S$ is a strict local martingale if and only if $\lim_{t\uparrow\infty}\EE^{\QQ}(S_t) = 0$.
Since~$S$ is non-negative, this is equivalent to~$S_t$ converging to zero almost surely,
and Theorem~\ref{thm:Tehranchi} applies.

\subsection{Strict local martingales on the real line}
We have so far assumed the framework of a strict local martingale taking values on the non-negative half line.
One could in principle consider local martingales taking values on the whole real line.
The motivation for such a consideration is the unique strong solution 
(on some given probability space $(\Omega, \Ff, (\Ff_t)_{t\geq 0}, \QQ)$)
to the stochastic differential equation
$\D S_t = \Pp(S_t)\D W_t$, starting at $S_0\in \RR$, 
where $\Pp(s)\equiv a s^2 + bs + c$ is a second-order polynomial.
This class of models was first proposed by Rady~\cite{Rady}, and then studied further 
in~\cite{Andersen, Lipton, Zuhlsdorff} and references therein.
It is clear that this class encompasses the standard Brownian motion ($a=b=0$),
the geometric Brownian motion ($a=c=0$) and the reciprocal of the three-dimensional Bessel process
($b=c=0$, see also Section~\ref{sec:CEVExample} with $\beta=1$).
The process~$S$ is clearly a continuous local martingale, and the localisation of the roots of~$\Pp$ as well as the starting point~$S_0$ fully characterise its martingale properties.
Let $\tau$ represent the first hitting time of zero by~$S$, 
and $S^{\tau}:=S_{\cdot\wedge\tau}$ 
the corresponding stopped process.
We summarise the results of Andersen~\cite{Andersen} as follows:
\begin{enumerate}[(i)]
\item $S$ is a true $\QQ$-martingale if and only if
either $a=b=0$ or $\Pp$ has two real roots $l$ and $u$ satisfying $l\leq S_0\leq u$;
\item $S^\tau$ is a true martingale if and only if $a=0$ or $\Pp$ has a unique real root greater or equal than $S_0$;
\item if $\Pp$ has two real roots $l< u$, then
\begin{enumerate}[(a)]
\item if $S_0 >u$, $S$ is a strict supermartingale;
\item if $S_0 <l$, $S$ is a strict submartingale;
\end{enumerate}
\item if $\Pp$ has a single real root~$l=u$, $S$ is a strict supermartingale if $S_0>u$ and a strict submartingale
if $S_0<u$;
\item if $\Pp$ has no real root, then $S$ is an unbounded strict local martingale.
\end{enumerate}
These results are intuitively clear since any real root of~$\Pp$ acts as an absorbing boundary,
and since a local martingale is a supermartingale when bounded below, 
and a submartingale when bounded above.
Andersen~\cite{Andersen} further provides pricing formulae for European Call and Put option prices 
in each of these cases.
Let us look at Case~(iii) more carefully.
In Case~(iii)(a), since the stock price is bounded below by $u$, the Put price (with strike $\E^{x}>0$) is clearly bounded between zero and $(\E^{x}-u)_+$, and is therefore a true $\QQ$-martingale.
This is the framework we have considered so far (with $u=0$).
Clearly, if the strike $\E^{x}$ lies between zero and $u$, then the Put price is always equal to zero.
Symmetrically, consider Case (iii)(b);
in this case, the stock price is bounded above by $l$, and therefore the Call price is bounded 
and therefore a true $\QQ$-martingale.
We can gather these results (and translate them into implied volatilities) in the following proposition:
\begin{proposition}
Let $S$ be a strict local martingale on $(\Omega, \Ff, (\Ff_t)_{t\geq 0}, \QQ)$ 
whose distribution is supported on some interval~$\Ii\subset\RR$.
Let $l\leq u$ be two real numbers, and recall that $x^*(\alpha)$ is defined in Theorem~\ref{thm:ExistenceIV}.
\begin{enumerate}[(i)]
\item if $\Ii=[l, u]$, then the implied volatility is always well defined;
\item If $\Ii = [u, +\infty)$, then the Put-implied volatility is always well defined;
for any $\alpha \in [0,1]$, the collateralised Call implied volatility is only defined on $[x^*(\alpha), +\infty)$;
\item If $\Ii = (-\infty, l]$, then the (fully collateralised) Call-implied volatility is always well defined;
furthermore, there exists $\underline{x}>0$ such that the Put-implied volatility is not well defined
for $x\leq \underline{x}$, but is so on $(\underline{x}, +\infty)$.
\end{enumerate}
Whenever $\E^{x}\notin \Ii$, the implied volatility, when it exists, is null;
the intervals can be taken open instead of closed.
\end{proposition}
\begin{proof}
In case (i), both European Call and Put prices are bounded, so that they are true martingales, 
and the result is straightforward.
Case (ii): the stock price is a strict supermartingale, which is the framework analysed above,
and the result follows from Theorem~\ref{thm:ExistenceIV}.
Case (iii): in this case, the stock price is a strict submartingale.
The Call price is bounded, hence a true martingale, and the Call-implied volatility is always well defined.
Clearly the Put is unbounded, so is not a true martingale.
The proof is symmetric to the Call case in Theorem~\ref{thm:ExistenceIV} and therefore omitted.
\end{proof}

\section{Examples}

\subsection{Strict local martingale with deterministic endpoint}\label{ex:bridge}
Let $T>0$, $\mu \ge 0$, $W$ be a standard Brownian motion, and $(M_t)_{t\geq 0}$ 
be the unique strong solution to the stochastic differential equation
$$
\D M_t = \frac{M_t - \mu}{\sqrt{T-t}}\D W_t
\qquad M_0 = 1.
$$
It is immediate to see that $M_t = (1 - \mu) \exp\left(\widetilde{W}_{\phi_t} - \frac{1}{2}\phi_t\right) + \mu$, 
where~$\widetilde{W}$ is standard Brownian motion and 
$\phi_t \equiv -\log\left(1 - t/T\right)$.
Then, for $\mu \neq 1$, $M$ is a non-negative strict local martingale on $[0,T]$ with $M_T=\mu$ almost surely.
This example shows attainment of lower bounds in~\eqref{eq:call_bounds} and~\eqref{eq:put_bounds}. 
Setting $\mu = 0$ we can create a strict local martingale with mass at zero at time $T$, 
i.e., examples of processes that do not fall within the scope of the duality approach.

\subsection{The CEV model}\label{sec:CEVExample}
On a probability space $(\Omega, \Ff, (\Ff_t)_{t\geq 0}, \QQ)$ supporting a Brownian motion~$W$, 
the stochastic differential equation
$\D S_t = \sigma S_t^{1+\beta}\D W_t$, 
with $S_0>0$,
admits a unique strong solution, which is a true martingale if and only if $\beta<0$ 
(see~\cite[Chapter 6.4]{JCY}).
In the case $\beta>0$, $S$ is a strict local martingale and the process $Y:=S^{-\beta}/(\sigma\beta)$
is a Bessel process of index $1/(2\beta)$ (equivalently of dimension $2+1/\beta$).
By the Feller classification of boundary points~\cite[Chapter 15, Section 6]{KarlinTaylor},
the origin is unattainable in finite time, and for any $t>0$, the transition density of~$S_t$ reads
$$
\PP(S_t \in \D s) = \frac{S_0^{1/2}s^{-2\beta-3/2}}{\sigma^2\beta t}
\exp\left(-\frac{S_0^{-2\beta} + s^{-2\beta}}{2\sigma^2\beta^2 t}\right)
I_{1/(2\beta)}\left(\frac{(S_0 s)^{-\beta}}{\sigma^2 \beta^2 t}\right)\D s,
$$
where $I_\cdot(\cdot)$ is the modified Bessel function of the first kind.
The expectation can be computed in closed form as
$\EE^{\QQ}(S_t) = S_0\Gamma\left(-\frac{1}{2\beta}, \frac{S_0^{-2\beta}}{2\sigma^2 \beta^2 t}\right)$,
where $\Gamma(a,x):=\Gamma(a)^{-1}\int_{0}^{x} u^{a-1}\E^{-u}\D u$ is the (normalised)
incomplete Gamma function.
In the case $\beta=1$, the density simplifies to
\begin{equation}\label{eq:densityCEV2}
\PP\left(S_t \in \D s\right) = \frac{S_0\D s}{\sigma s^3\sqrt{2\pi t}}
\left\{\exp\left(-\frac{(1/s - 1/S_0)^2}{2t\sigma^2}\right) - \exp\left(-\frac{(1/s + 1/S_0)^2}{2t \sigma^2}\right)\right\},
\end{equation}
and $\EE(S_T) = S_0(1 - 2\Nn(-1/(S_0\sqrt{T})))$.
In Figure~\ref{fig:CEV}, we illustrate numerically Theorem~\ref{thm:IVExpansion}.
Regarding the large-time behaviour of the Put implied volatility smile,
Theorem~\ref{thm:Tehranchi} yields (see also Tehranchi~\cite[Example 5.9]{Tehranchi})
$$
I_S(x,T)^2 T = 4\log(T) - 4\log\log(T) + 4x + \varepsilon(x,T),
$$
for all $x\in\RR$ as $T$ tends to infinity.
\begin{figure}[h!tb] 
\centering
\includegraphics[scale=0.4]{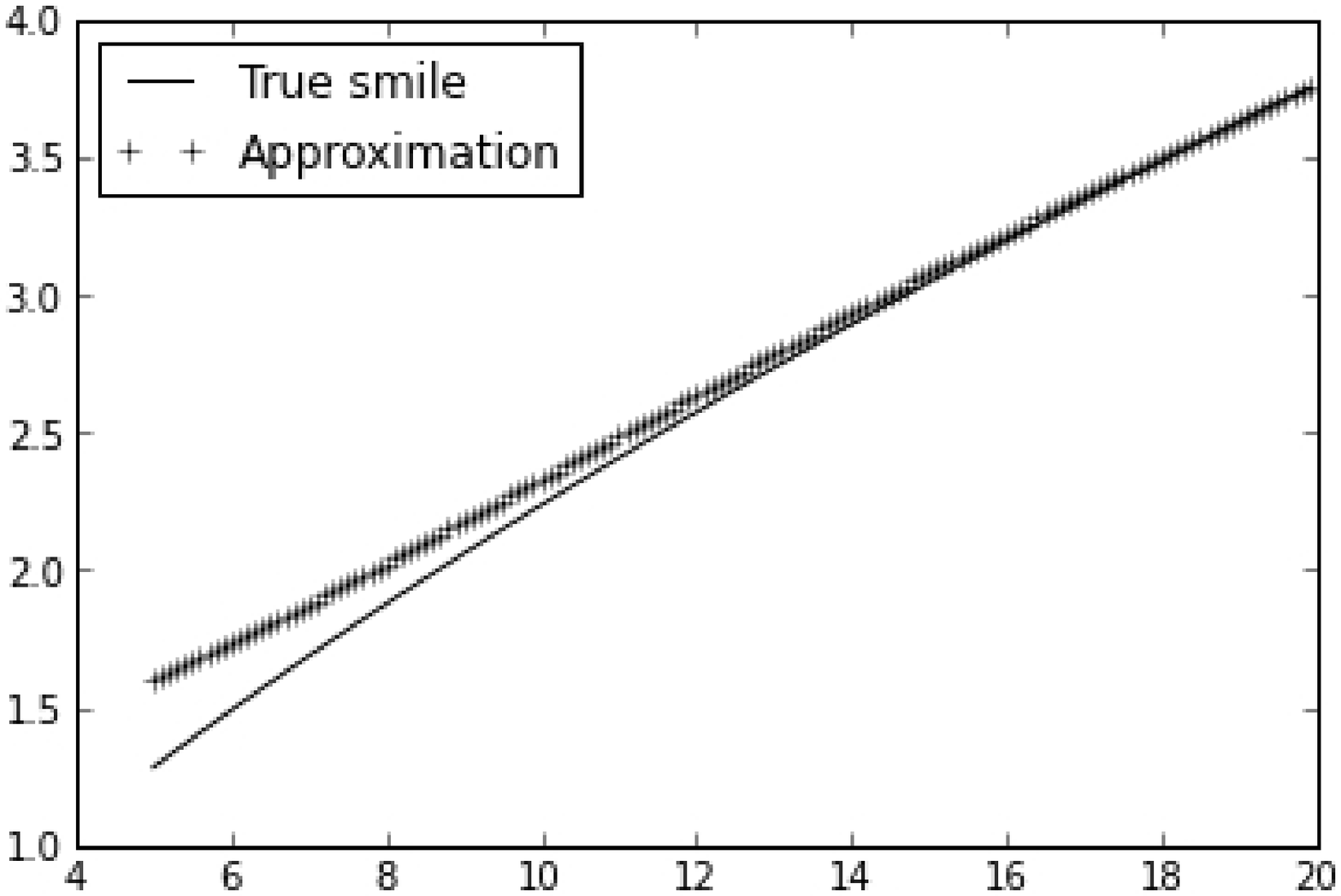}
\includegraphics[scale=0.4]{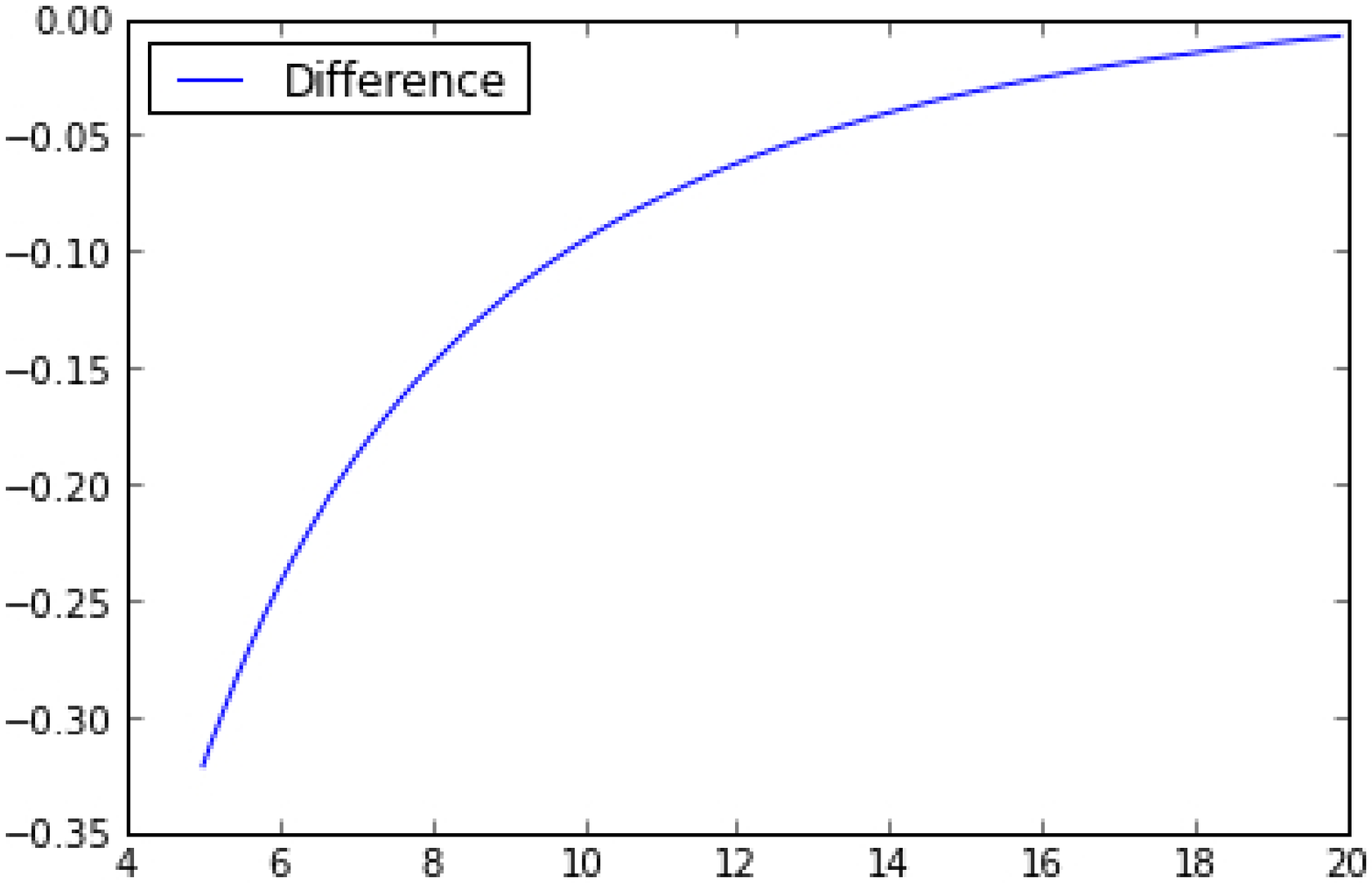}
\caption{
We take $(S_0, \beta, \sigma, T) = (1, 2.4, 10\%, 1)$. 
The horizontal axis represents the log strikes;
the left figures represent the true value of $x\mapsto I_s^c(x)$ (solid line)
and its approximation from Theorem~\ref{thm:IVExpansion} (crosses).
The right graph represents the error between the true value and its approximation.
}
\label{fig:CEV}
\end{figure}


\end{document}